\documentclass[9pt,twocolumn,twoside]{pnas-new}

\templatetype{pnasresearcharticle} %

\setboolean{displaywatermark}{false}

\usepackage{amsmath}
\usepackage{amssymb}
\usepackage{enumitem}
\usepackage{tabularx}
\usepackage{xr}

\makeatletter
\newcommand*{\addFileDependency}[1]{%
    \typeout{(#1)}
    \@addtofilelist{#1}
    \IfFileExists{#1}{}{\typeout{No file #1.}}
}
\makeatother

\newtheorem{theorem}{Theorem}[section]
\newtheorem{lemma}[theorem]{Lemma}

\newenvironment{theorem*}[1][]{%
    \renewcommand{\thetheorem}{#1}
    \begin{theorem}
}{%
    \end{theorem}
    \addtocounter{theorem}{-1}
}

\newcommand{\qedsymbol}{\mbox{$\Box$}}
\newcommand{\qed}{\hfill\qedsymbol}

\newenvironment{proof}[1][Proof]{%
    \begin{trivlist}\item[\hskip\labelsep{\bfseries #1}]%
}{%
    \qed\end{trivlist}%
}

\newcommand{\amp}{\mathbin{\&}}
\renewcommand{\phi}{\varphi}
\renewcommand{\models}{\vDash}
\newcommand{\infers}{\vdash}
\newcommand{\alphaop}{\mathop{\alpha}}

\title{%
    Foundations of Reasoning with Uncertainty via Real-valued Logics
}

\author[a,1]{Ronald Fagin}
\author[a]{Ryan Riegel}
\author[a]{Alexander Gray}

\affil[a]{IBM Research}

\leadauthor{Fagin} 

\significancestatement{%
    This work is a step in the direction of explainable AI as it pertains to logical inference in neural networks.
    This may ultimately assist in preventing unfair, unwarranted, or otherwise undesirable outcomes from the application of modern AI methods.
}

\correspondingauthor{\textsuperscript{1}To whom correspondence should be addressed. E-mail: fagin@us.ibm.com}


\begin{abstract}
    Real-valued logics underlie an increasing number of neuro-symbolic approaches, though typically their logical inference capabilities are characterized only qualitatively.
    We provide foundations for establishing the correctness and power of such systems.
    We give a sound and strongly complete axiomatization that can be parametrized to cover essentially every real-valued logic, including all the common fuzzy logics. 
Our class of sentences are very rich, and each describes a set of possible real values for a collection of formulas of the real-valued logic, including which combinations of real values are possible. 
Strong completeness 
allows us to derive exactly what information can be inferred about the combinations of real values of a collection of formulas given information about the combinations of real values of several other collections of formulas.
    We then extend the axiomatization to deal with weighted subformulas.
    Finally, we give a decision procedure based on linear programming for deciding, for certain real-valued logics and under certain natural assumptions, whether a set of our sentences logically implies another of our sentences.
\end{abstract}

\dates{This manuscript was compiled on \today}
\doi{\url{www.pnas.org/cgi/doi/10.1073/pnas.XXXXXXXXXX}}

\begin{document}

\maketitle
\thispagestyle{firststyle}
\ifthenelse{\boolean{shortarticle}}{\ifthenelse{\boolean{singlecolumn}}{\abscontentformatted}{\abscontent}}{}

\dropcap{R}ecent years have seen 
growing interest in approaches for
augmenting the capabilities of learning-based methods with those of reasoning, often broadly referred to as \emph{neuro-symbolic} (though they may not be strictly neural).
One of the key goals that neuro-symbolic approaches have at their root is logical inference, or reasoning.
However, the representation of classical 0-1 logic (where truth values of sentences are either 0, representing ``False'', or 1, representing ``True'') is
generally insufficient for this goal because representing uncertainty is essential to AI.
In order to merge with the ideas of neural learning, the truth values dealt with must be \emph{real-valued} (we shall take these to be real numbers in the interval \([0, 1]\), where intuitively, 0 means ``completely false'', and 1 means ``completely true''), whether the uncertainty semantics are those of probabilities, subjective beliefs, neural network activations, or fuzzy set memberships.
For this reason, many major approaches have turned to real-valued logics.
Logic tensor networks \cite{serafini2016logic} define a logical language on real-valued vectors corresponding to groundings of terms computed by a neural network, which can use any of the common real-valued logics (e.g., {\L}ukasiewicz, product, or G\"odel logic) for its connectives (e.g., \(\amp\), \(\veebar\), \(\neg\), and \(\rightarrow\)).
Probabilistic soft logics \cite{bach2017hinge} draw a correspondence of their approach based on Markov random fields (MRFs) with satisfiability of statements in a real-valued logic ({\L}ukasiewicz).
Tensorlog \cite{cohen2020tensorlog}, also based on MRFs but implemented in neural network frameworks, draws a correspondence of its approach to the use of connectives in a real-valued logic (product).
Logical neural networks (LNNs) \cite{riegel2020logical} draw a correspondence between activation functions of neural networks and connectives in real-valued logics.
To complete a full correspondence between neural networks and statements in real-valued logic, LNN defines a class of real-valued logics allowing weighted inputs, which represent the relative influence of subformulas.
While widely regarded as fundamental to the goal of AI, the reasoning capabilities of the aforementioned systems are typically made qualitatively versus quantitatively and mathematically.
While learning theory (roughly, what it means to perform learning) is well articulated and, for 0-1 logic, what it means to perform reasoning is well studied, reasoning is surprisingly not well formalized for real-valued logics.
As reasoning becomes an increasing goal of learning-based work, it becomes important to have a solid mathematical footing for it.

Formalization of the idea of real-valued logics is old and fundamental, going back to the origins of formal logic.
It is not well known that Boole himself invented a probabilistic logic in the 19th century \cite{boole1854investigation}, where formulas were assigned real values corresponding to probabilities.
It was used in AI to model the semantics of vague concepts for commonsense reasoning by expert systems \cite{zadeh1975fuzzy}.
Real-valued logic is used in linguistics to model certain natural language phenomena \cite{novak2008formal}, in hardware design to deal with multiple stable voltage levels \cite{epstein1993multiple}, and in databases to deal with queries that are composed of multiple graded notions, such as the redness of an object, that can range from 0 (``not at all red'') to 1 (``completely red'') \cite{fagin2003optimal}.
Despite all this, while definitions of logical correctness and power (generally, soundness and completeness) are well established and corresponding procedures for \emph{theorem proving} having those properties are abundant for classical logics, 
the equivalents for real-valued logics (where the values can take arbitrarily values between 0 and 1) are rather limited. 

\paragraph{This paper.}
In this paper, there are two levels of logic.
In the ``inner'' layer, we have formulas of the real-valued logic with its logical connectives.
In this inner layer, we shall use $\amp$ for ``and'' and $\veebar$ for  ``or'', as is done in \cite{hajek1998metamathematics}.
In the ``outer'' layer, we have a novel class of sentences about the inner real-valued logic (such as saying which truth values a given real-valued formula may attain).
For these sentences (which take on only the classical values 0 and 1 for False and True, respectively), we make use of the traditional logical symbols $\land$ for ``and'' and $\lor$ for ``or''.
We remark that, somewhat confusingly, the symbols $\land$ and $\lor$ are often used in real-valued logics for weaker versions of ``and'' and ``or'' than that given by $\amp$ and $\veebar$, which we do not have need to discuss in this paper.

Let us say that an axiomatization of a logic is \emph{strongly complete} if whenever $\Gamma$ is a finite set of sentences in the (outer) logic and $\gamma$ is a single sentence in the (outer) logic that is a logical consequence of $\Gamma$, then there is a proof of $\gamma$ from $\Gamma$ using the axiomatization.
An axiomatization is \emph{weakly complete} if this holds for $\Gamma = \emptyset$.
That is, an axiomatization is weakly complete if whenever $\gamma$ is a valid sentence (always true), then there is a proof of $\gamma$ using the axiomatization.
Early axiomatizations of real-valued logics in the literature were typically weakly complete, but now have been improved to strongly complete 
(see \cite{ hajek1998metamathematics} for examples).

We now explain why it is necessary to assume that $\Gamma$ is finite in the definition of strong completeness.
(In our explanation, we make use of ideas from \cite {hajek1998metamathematics}.)
Let us restrict to {\L}ukasiewicz logic. 
Let $A^k$ denote $A \amp A \amp \cdots \amp A$, where $A$ appears $k$ times.
Let $\Gamma$ be the infinite set of sentences $(B \rightarrow A^k; \{1\})$ for $k \geq 1$, along with $(A; [0,1))$ which says that the value of $A$ is less than 1. 
Let $\tau$ be $(B; \{0\})$.
We now show that $\Gamma$ logically implies $\tau$.
Assume that $\Gamma$ holds but $\tau$ does not hold. Therefore, the value of $A$ is less than 1.
It then follows from the definition of conjunction in
{\L}ukasiewicz logic that  there is $k$ such that $A^k$ has value 0 .
From $(B \rightarrow A^k; \{1\})$ this then implies that the value of $B$ is 0, so $\tau$ holds. 
Hence,  $\Gamma$ logically implies $\tau$.  Because our proofs are of finite length, there cannot be a proof of $\tau$ from $\Gamma$, since this would give a proof of $\tau$ from a finite subset of $\Gamma$, but no finite subset of $\Gamma$ logically implies $\tau$.
A natural open problem is whether we can allow $\Gamma$ to be infinite if we were to restrict our attention to G\"odel logic.

We introduce a rich, novel class of sentences.
\begin{enumerate}
\item
These sentences can say what the set $S$ of possible values is for a formula $\sigma$.
This set $S$ can be a singleton $\{s\}$ (meaning that the real value of $\sigma$ is $s$), or $S$ can be an interval, or a union of intervals, or in fact an arbitrary subset of $[0,1]$, e.g.  the set of rational numbers in $[0,1]$.
\item
Our sentences can give not only the possible real values of formulas, but the interactions between these values. 
For example, if $\sigma_1$ and $\sigma_2$ are formulas, our sentences can not only say what the possible real values are for each of $\sigma_1$ and $\sigma_2$, but also how they interact:  thus, if $s_1$ is the real value of $\sigma_1$ and $s_2$ is the real value of $\sigma_2$, then there is a sentence in our logic that says $(s_1,s_2)$ must lie in the set $S$ of ordered pairs, where $S$ is an arbitrary subset of $[0,1] \times [0,1]$.
We give a sound and strongly complete axiomatization for our sentences.
\item
Unlike the other axiomatizations mentioned earlier, our axiomatization can be extended to include the use of weights for subformulas (where, for example,in the formulas $A_1 \veebar A_2$, the subformula $A_1$ is considered twice as important as the subformula $A_2$). 
\item
A surprising feature of our axiomatization is that it is parametrized, so that this one axiomatization is sound and strongly complete for essentially \emph{every} real-valued logic, including those that do not obey the standard restrictions on fuzzy logics (such as conjunction being commutative). 
Previous axiomatizations in the literature had a separate set of axioms for each real-valued logic (for example, one of the axioms for {\L}ukasiewicz logic is $A \leftrightarrow \neg \neg A$, and one of the axioms for G\"odel logic is $A \leftrightarrow (A \amp A)$).
In the axiomatizations mentioned earlier, each connective has a fixed associated function that tells how to evaluate it.
For example, in in {\L}ukasiewicz logic,
the value of $A_1 \amp A_2$ is $f_{\amp}(a_1,a_2)$, where $a_1$ is the value of $A_1$ and $a_2$ is the value of $A_2$, and where 
$f_{\amp}(a_1,a_2) = \max \{0, a_1 + a_2 - 1\}$.
By contrast, for our axiomatization, $f_{\amp}$ is arbitrary, as long as it maps $[0,1]^2$ into $[0,1]$.

\end{enumerate}

From now on (except in the Section~\ref{sec:related}  on related work) we use ``complete'' to mean ``strongly complete''.

An especially useful real-valued logic for logical neural nets is {\L}ukasiewicz logic, for several reasons.
First, the \(\amp\), \(\veebar\), \(\neg\), and \(\rightarrow\) operators are essentially linear, in that if \(a_1\) is the truth value of a formula \(A_1\), and \(a_2\) is the truth value of a formula \(A_2\), then (a) \(A_1 \amp A_2\) has value \(\max \{0, a_1 + a_2 - 1\}\), (b) \(A_1 \veebar A_2\) has value \(\min\{1, a_1 + a_2\}\), (c) \(\neg A_1\) has value \(1 - a_1\), and (d) \(A_1 \rightarrow A_2\) is equivalent to \(\neg A_1 \veebar A_2\), and so has value \(\min\{1, 1 - a_1 + a_2\}\).\footnote{%
    The versions of \(\amp\) and \(\veebar\) we describe here are sometimes called in the literature \emph{strong conjunction} and \emph{strong disjunction}.
    \emph{Weak conjunction} is given by \(\min\{a_1, a_2\}\) and \emph{weak disjunction} is given by \(\max\{a_1, a_2\}\).
}\label{foot}
Second, it is easy to incorporate weights.
Thus, if \(w_1\) and \(w_2\) are nonnegative weights of \(A_1\) and \(A_2\), respectively, then we can take the weighted value of \(A_1 \veebar A_2\) to be %
\(\min\{1, w_1 s_1 + w_2 s_2\}\).

Throughout this paper, we take the domain of each function in the real-valued logic to be \([0, 1]\) or \([0, 1] \times [0, 1]\) and the range to be \([0, 1]\).
This is a common assumption for many real-valued logics, but all of our results go through with obvious modifications if the domains are \(D^k\) for possibly multiple choices of arity \(k\) and range \(D\), for arbitrary subsets \(D\) of the reals.
We note that real-valued logic can be viewed as a special case of multi-valued logic \cite{rescher1969many}, although in multi-valued logic there is typically a finite set of truth values, not necessarily linearly ordered.

We also provide a decision procedure for deciding, whether a set of our sentences logically implies another of our sentences, for certain real-valued logics, under certain natural assumptions.
We implemented the decision algorithm, dubbed \textbf{SoCRAtic} logic (for \textbf{So}und and \textbf{C}omplete \textbf{R}eal-valued \textbf{A}xioma\textbf{tic} logic), which we describe in detail and make available in source code.

Our sentences allow arbitrary real-valued logics, as does our sound and complete axiomatization, but our decision procedure depends heavily on the choice of real-valued logic, and in particular is tailored towards {\L}ukasiewicz and G\"odel logic.
This is because a key portion of our decision procedure is linear programming, and we depend on the essentially linear nature of {\L}ukasiewicz logic and the ease of dealing with min and max in G\"odel logic.

\paragraph{Overview.}
Until the final section, we do not allow weights.
In Section~\ref{sec:models}, we give our basic notions, including what a model is and what a sentence is.
In Section~\ref{sec:models}, we define our sentences to be of the form \((\sigma_1, \ldots, \sigma_k; S)\) where the \(\sigma_i\) are formulas, where \(S\) is a set of tuples \((s_1, \ldots, s_k)\), and where the sentence says that if the value of each \(\sigma_i\) is \(s_i\), for \(1 \leq i \leq k\), then \((s_1, \ldots, s_k) \in S\).
In Section~\ref{sec:axioms}, we give our (only) axiom and our inference rules.
In Section~\ref{sec:completeness}, we give our soundness and completeness theorem.
In Section~\ref{sec:Boolean}, we give a theorem that says that our sentences are closed under Boolean combinations.
This helps show robustness of our class of sentences.
In Section~\ref{sec:dimension} we discuss possible simplifications of our sentences.
In Section~\ref{sec:decision}, we give the decision algorithm.
In Section~\ref{sec:weights}, we show how to extend our methodology to incorporate weights.
In Section~\ref{sec:related}, we discuss related work. 
In the Conclusions, we review the implications for neuro-symbolic approaches.

\section{Models and sentences}\label{sec:models}

We assume a finite set of atomic propositions.
These can be thought of as the leaves of a neural net, i.e., nodes with no inputs from other neurons.
A model \(M\) is an assignment \(g^M\) of values to the atomic propositions.
Thus, \(M\) assigns a value \(g^M(A) \in [0, 1]\) to each atomic proposition \(A\).

Let \(F\) be the set of logical formulas over the atomic propositions, where we allow arbitrary finite sets of binary and unary connectives.
Typical binary connectives are conjunction (denoted by \(\amp\)), disjunction (denoted by \(\veebar\)), and implication (denoted by \(\rightarrow\)).
Typical unary connectives are negation (denoted by \(\neg\)) and a delta function (denoted by \(\bigtriangleup\)).
Sometimes \(\neg x\) is taken to be \(1 - x\), and \(\bigtriangleup x\) is taken to be defined by \(\bigtriangleup x = 1\) if \(x = 1\) and 0 otherwise.

When considering only formulas with value 1, as most other works do when giving sound axiomatizations of real-valued logics, the convention is to consider a sentence to be simply a member of \(F\).
What if we want to take into account values other than 1?
We take a sentence to be an expression of the form \((\sigma_1, \ldots, \sigma_k; S)\), where \(\sigma_1, \ldots, \sigma_k\) are in \(F\), and where \(S \subseteq [0, 1]^k\).
The intuition is that \((\sigma_1, \ldots, \sigma_k; S)\) says that if the value of each \(\sigma_i\) is \(s_i\), for \(1 \leq i \leq k\), then \((s_1, \ldots, s_k) \in S\).
We refer to our sentences as \emph{multi-dimensional sentences}, or for short \emph{MD-sentences}. 
For a fixed $k$, we refer to the MD-sentence
\((\sigma_1, \ldots, \sigma_k; S)\) as \emph{$k$-dimensional}.
The class of MD-sentences is robust.
In particular, Theorem~\ref{thm:robust} says that MD-sentences are closed under Boolean combinations.
We give a sound and (strongly) complete axiomatization, that is parameterized to deal with an arbitrary fixed real-valued logic.
This axiomatization allows us to derive exactly what information can be inferred about the combinations of values of a collection of formulas given information about the combinations of values of several other collections of formulas.

Note that we are not saying that the {\em logic} is multi-dimensional (which could mean that the values taken on by variables are vectors, not just numbers), but instead we are saying that the {\em sentences} in our "outer" logic are multi-dimensional. The "inner" logic we work with in this paper is real-valued, and real-valued logic has been heavily studied.  What is novel in our paper are our multi-dimensional sentences.

Note that the set $S$ in \((\sigma_1, \ldots, \sigma_k; S)\) can be undecidable, even if $k=1$ and every member of $S$ is a rational number.   For example, we could then take $S$  to  be the set of all numbers $1/k$, where $k$ is the G$\ddot{o}$del number of a halting Turing machine. But our decision procedures involve only special sets $S$.  Thus, we shall say in Section~\ref {sec:decision} that a sentence \((\sigma_1, \ldots, \sigma_k; S)\) is \emph{interval-based} if \(S\) is of the form \(S_1 \times \cdots \times S_k\), where each \(S_i\) is the union of a finite number of intervals with rational endpoints. And our decision procedure in that section deals with interval-based sentences.  However, our sound and complete axiomatization in Section~\ref {sec:completeness} makes no such assumptions about the sets $S$; in particular, the sets $S$ can be undecidable.

For convenience, we assume throughout that in the sentence \((\sigma_1, \ldots, \sigma_k; S)\), we have that \(\sigma_i\) and \(\sigma_j\) are different formulas if \(i \neq j\).
We refer to \(\sigma_1, \ldots, \sigma_k\) as the \emph{components} of \((\sigma_1, \ldots, \sigma_k; S)\), and $S$ as the \emph{information set} of \((\sigma_1, \ldots, \sigma_k; S)\).

Let \(\gamma\) be the sentence \((\sigma_1, \ldots, \sigma_k; S)\).
We now say what it means for a model \(M\) to satisfy \(\gamma\).
For \(1 \leq i \leq k\), let \(s_i\) be the value of the formula \(\sigma_i\) under the assignment of values to the atomic propositions given by the model \(M\).
We say that \(M\) \emph{satisfies} \(\gamma\) if \((s_1, \ldots, s_k) \in S\).
We then say that \(M\) is a \emph{model} of \(\gamma\), and we write \(M \models \gamma\).
Note that if \((\sigma_1, \ldots, \sigma_k; S)\) is satisfiable, that is, has a model, then \(S \neq \emptyset\).

\section{Axioms and inference rules}\label{sec:axioms}

We have only one axiom:
\vspace{-0.2in}
\begin{equation}\label{in01}
    (\sigma; [0, 1])
\end{equation}
\vspace{-0.2in}

Axiom~(\ref{in01}) guarantees that all values are in \([0, 1])\).

We now give our inference rules.

If \(\pi\) is a permutation of \(1, \ldots, k\), then:
\vspace{-0.1in}
\begin{equation}\label{perm}
    \mbox{From } (\sigma_1, \ldots, \sigma_k; S) 
    \mbox{ infer } (\sigma_{\pi(1)}, \ldots, \sigma_{\pi(k)}; S')
\end{equation}
\vspace{-0.2in}

where \(S' = \{(s_{\pi(1)}, \ldots, s_{\pi(k)})\colon (s_1, \ldots, s_k) \in S\}\).

Rule~(\ref{perm}) simply permutes the order of the components.

Our next inference rule is:
\vspace{-0.1in}
\begin{equation}\label{add}
    \mbox{From } (\sigma_1, \ldots, \sigma_k; S) 
    \mbox{ infer }  
    \end{equation}
     $(\sigma_1, \ldots, \sigma_k, \sigma_{k+1}, \ldots, \sigma_m; S \times [0, 1]^{m-k}).$

Rule~(\ref{add}) extends \((\sigma_1, \ldots, \sigma_k; S)\) to include \(\sigma_{k+1}, \ldots, \sigma_m\) with no nontrivial information being given about 
the new components.

Our next inference rule is:
\vspace{-0.1in}
\begin{equation}\label{intersect}
    \mbox{From } (\sigma_1, \ldots, \sigma_k; S_1) 
    \mbox{ and } (\sigma_1, \ldots, \sigma_k; S_2) 
    \mbox{ infer }
    \end{equation}
    $(\sigma_1, \ldots, \sigma_k; S_1 \cap S_2).$

Rule~(\ref{intersect}) enables us to join the information in 
\((\sigma_1, \ldots, \sigma_k; S_1)\) 
and \((\sigma_1, \ldots, \sigma_k; S_2)\).

Our next inference rule is the following (where \(0 < r < k\)):
\vspace{-0.1in}
\begin{equation}\label{proj}
    \mbox{From } (\sigma_1, \ldots, \sigma_k; S) 
    \mbox{ infer } (\sigma_1, \ldots, \sigma_{k-r}; S')
\end{equation}
where \(S' = \{(s_1, \ldots, s_{k-r})\colon (s_1, \ldots, s_k) \in S\}\).

Intuitively, \(S'\) is the projection of \(S\) onto the first \(k-r\) components.
Rule~(\ref{proj}) enables us to select information about \(\sigma_1, \ldots, \sigma_{k-r}\) from information about \(\sigma_1, \ldots, \sigma_k\).

Our next inference rule is:
\vspace{-0.1in}
\begin{equation}\label{superset}
    \mbox{From } (\sigma_1, \ldots, \sigma_k; S) 
    \mbox{ infer } (\sigma_1, \ldots, \sigma_k; S') 
    \mbox{ if } S \subseteq S'.
\end{equation}

Rule~(\ref{superset}) says that we can go from more information to less information.
The intuition is that smaller information sets are more informative.

We now give an inference rule that depends on the real-valued logic under consideration.
For each real-valued binary connective \(\alpha\), let \(f_{\alpha}(s_1, s_2)\) be the value of \(\sigma_1 \alphaop \sigma_2\) when the value of \(\sigma_1\) is \(s_1\) and the value of \(\sigma_2\) is \(s_2\).
For example, in G\"odel logic, \(f_{\amp}(s_1, s_2) = \min\{s_1, s_2\}\).
For each real-valued unary connective \(\rho\), let \(f_{\rho}(s)\) be the value of \(\rho\sigma\) when the value of \(\sigma\) is \(s\).
For example, in {\L}ukasiewicz logic, \(f_{\neg}(s) = 1 - s\).

In the sentence $(\sigma_1, \ldots, \sigma_k; S)$,
let us say that the tuple $(s_1, \ldots, s_k)$ in $S$ is \emph{good} if (a) \(s_m = f_{\alpha} (s_i, s_j)\) whenever \(\sigma_m\) is \(\sigma_i \alphaop \sigma_j\) and \(\alpha\) is a binary connective (such as \(\amp\)), and (b) \(s_j = f_{\rho}(s_i)\) whenever \(\sigma_j\) is \(\rho\sigma_i\) and \(\rho\) is a unary connective (such as \(\neg\)).
Note that being ``good'' is a local property of a tuple $s$ in $S$ (that is, it depends only on the tuple $s$ and not on the other tuples in $S$).
Of course, if the real-valued logic under consideration has higher-order connectives (ternary, etc.), then we would modify the definition of a good tuple in the obvious way.
For simplicity, we will assume throughout this paper that we are in the common case where the only connectives of the real-valued logic are unary and binary, although all of our results go through in the general case.

We then have the following inference rule:
\vspace{-0.1in}
\begin{equation}\label{operators}
    \mbox{From } (\sigma_1, \ldots, \sigma_k; S) 
    \mbox{ infer } (\sigma_1, \ldots, \sigma_k; S')
\end{equation}
when $S'$ is the set of good tuples of $S$.

Rule~(\ref{operators}) is our key rule of inference.
Let \(\gamma_1\) be the premise \((\sigma_1, \ldots, \sigma_k; S)\) and let \(\gamma_2\) be the conclusion \((\sigma_1, \ldots, \sigma_k; S')\) of Rule~(\ref{operators}).
As we shall discuss later, \(\gamma_1\) and \(\gamma_2\) are logically equivalent, and \(S'\) is as small as possible so that \(\gamma_1\) and \(\gamma_2\) are logically equivalent.

A simple example of a valid sentence is $(A,B,A \veebar B; S)$ where 
$S = \{(s_1,s_2,s_3) \colon  \ s_1 \in [0,1], s_2 \in [0,11], s_3 = f_{\veebar}(s_1,s_2)\}$.
This is derived from the valid sentence $(A,B,A \veebar B; [0,1]^3)$ by applying Rule~(\ref{operators})

Each of our rules is of the form ``From A infer B'' or ``From A infer B where ...''.
We refer to A as the \emph{premise} and B as the \emph{conclusion}.
We need the notion of a \emph{subformula} of a formula.
If \(\alpha\) is a binary connective, then the subformulas of \(\sigma_1 \alphaop \sigma_2\) are \(\sigma_1\) and \(\sigma_2\).
If \(\rho\) is a unary connective, then the subformula of \(\rho\sigma\) is \(\sigma\).

Let \(\Gamma\) be a set of MD-sentences.
We define the \emph{closure \(G\) of\/ \(\Gamma\) under subformulas} as follows.
For each sentence $(\gamma_1, \ldots, \gamma_m; S)$ in $\Gamma$, 
the set \(G\) contains $\gamma_1, \ldots, \gamma_m$, 
and for each formula $\gamma$ in $G$, the set $G$ contains every subformula of $\gamma$.

In particular, \(G\) contains every atomic proposition that appears inside the components of \(\Gamma\).

\section{Soundness and completeness}\label{sec:completeness}

Let \(\Gamma\) be a finite set of MD-sentences, and let \(\gamma\) be a single MD-sentence.
We write \(\Gamma \models \gamma\) if every model of \(\Gamma\) is a model of \(\gamma\).
We write \(\Gamma \infers \gamma\) if there is a proof of \(\gamma\) from \(\Gamma\), using our axiom system.
\emph{Soundness} says ``\(\Gamma \infers \gamma\) implies \(\Gamma \models \gamma\)''.
\emph{Completeness} says ``\(\Gamma \models \gamma\) implies \(\Gamma \infers \gamma\)''.
In this section, we shall prove that our axiom system is sound and complete for MD-sentences.
We define a special property of certain MD-sentences, that is used in a crucial manner in our completeness proof.
Let us say that a sentence \((\sigma_1, \ldots, \sigma_k; S)\) is \emph{minimized} if whenever \((s_1, \ldots, s_k) \in S\), then there is a model \(M\) of \((\sigma_1, \ldots, \sigma_k; S)\) such that for \(1 \leq i \leq k\), the value of \(\sigma_i\) in \(M\) is \(s_i\).
Thus, \((s_1, \ldots, s_k) \in S\) \emph{if and only if} there is a model \(M\) of \((\sigma_1, \ldots, \sigma_k; S)\) such that for \(1 \leq i \leq k\), the value of \(\sigma_i\) in \(M\) is \(s_i\).
We use the word ``minimized'', since intuitively, \(S\) is as small as possible.

Our proof of completeness makes use of the following lemmas.

\begin{lemma}\label{lem:min}
    Let \((\sigma_1, \ldots, \sigma_k; S)\) be the premise of Rule~(\ref{operators}).
    Assume that \(G = \{\sigma_1, \ldots, \sigma_k\}\) is closed under subformulas (so that in particular, every atomic proposition that appears in \(G\) is a member of \(G\)).
    Then the conclusion \((\sigma_1, \ldots, \sigma_k; S')\) of Rule~(\ref{operators}) is minimized.
\end{lemma}
\begin{proof}
    Let \(\phi\) be the conclusion \((\sigma_1, \ldots, \sigma_k; S')\) of Rule~(\ref{operators}).
    Assume that \((s_1, \ldots, s_k) \in S'\).
    To prove that \(\phi\) is minimized, we must show that there is a model \(M\) of \(\phi\) such that for \(1 \leq i \leq k\), the value of \(\sigma_i\) in \(M\) is \(s_i\).
    From the assignment of values to the atomic propositions, as specified by a portion of \((s_1, \ldots, s_k)\), we obtain our model \(M\).
    For this model \(M\), the value of each \(\sigma_i\) is exactly that specified by \((s_1, \ldots, s_k)\), as we can see by a simple induction on the structure of formulas.
    Hence, \(\phi\) is minimized.
\end{proof}

\begin{lemma}\label{lem:equiv}
    For each of Rules~(\ref{perm}), (\ref{add}), (\ref{intersect}), and (\ref{operators}), the premise and the conclusion are logically equivalent.
\end{lemma}
\begin{proof}
    The equivalence of the premise and conclusion of Rule~(\ref{perm}) is clear.
    For Rules~(\ref{add}), (\ref{intersect}), and (\ref{operators}),
    the fact that the premise logically implies the conclusion follows from soundness of the rules, which we shall show shortly.
    We now show that for Rules~(\ref{add}), (\ref{intersect}), and (\ref{operators}), the conclusion logically implies the premise.
    For Rule~(\ref{add}), we see that if \((s_1, \ldots, s_m) \in S \times [0, 1]^{m-k}\), then \((s_1, \ldots, s_k) \in S\).
    Hence, the conclusion of Rule~(\ref{add}) logically implies the premise of Rule~(\ref{add}).
    For Rules~(\ref{intersect}) and (\ref{operators}), the conclusion logically implies the premise because of the soundness of Rule~(\ref{superset}).
\end{proof}

\begin{lemma}\label{lem:moremin}
    Minimization is preserved by Rules~\ref{perm} and (\ref{intersect}),
    in the following sense.
    \begin{enumerate}
        \item If the premise of Rule~(\ref{perm}) is minimized, then so is the conclusion.
        \item If the premises \((\sigma_1, \ldots, \sigma_k; S_1)\) and \((\sigma_1, \ldots, \sigma_k; S_2)\) of Rule~(\ref{intersect}) are minimized, then so is the conclusion \((\sigma_1, \ldots, \sigma_k; S_1 \cap S_2)\).
    \end{enumerate}
\end{lemma}
\begin{proof}
    Part (1) is immediate, since the premise and conclusion have exactly the same information.
    
    For part (2), assume that \((\sigma_1, \ldots, \sigma_k; S_1)\) and \((\sigma_1, \ldots, \sigma_k; S_2)\) are minimized.
    To show that \((\sigma_1, \ldots, \sigma_k; S_1 \cap S_2)\) is minimized, we must show that if \((s_1, \ldots, s_k) \in S_1 \cap S_2\), then there is a model \(M\) of \((\sigma_1, \ldots, \sigma_k; S_1 \cap S_2)\) such that for \(1 \leq i \leq k\), the value of \(\sigma_i\) in \(M\) is \(s_i\).
    Assume that \((s_1, \ldots, s_k) \in S_1 \cap S_2\).
    Hence, \((s_1, \ldots, s_k) \in S_1\).
    Since \((\sigma_1, \ldots, \sigma_k; S_1)\) is minimized, we obtain the desired model \(M\).
\end{proof}

\begin{theorem}\label{thm:soundcomplete}
    Our axiom system is sound and complete for MD-sentences.
\end{theorem}
\begin{proof}
We begin by proving soundness.
We say that an axiom is sound if it is true in every model.
We say that an inference rule is sound if every model that satisfies the premise also satisfies the conclusion.
To prove soundness of our axiom system, it is sufficient to show that our axiom is sound and that each of our rules is sound.

Axiom~(\ref{in01}) is sound, since every real-valued logic formula has a value in \([0, 1]\).

Rule~(\ref{perm}) is sound, since the premise
and conclusion
encode exactly the same information.

Rule~(\ref{add}) is sound for the following reason.
Let \(M\) be a model, and let \(s_1, \ldots, s_m\) be the values of \(\sigma_1, \ldots, \sigma_m\), respectively, in \(M\).
If \(M\) satisfies the premise, then \((s_1, \ldots, s_k) \in S\).
This implies that \((s_1, \ldots, s_m ) \in S \times [0, 1]^{m-k})\) and so \(M\) satisfies the conclusion.

Rule~(\ref{intersect}) is sound for the following reason.
Let \(M\) be a model, and let \(s_1, \ldots, s_k\) be the values of \(\sigma_1, \ldots, \sigma_k\), respectively, in \(M\).
If \(M\) satisfies the premise, then \((s_1, \ldots, s_k) \in S_1\) and \((s_1, \ldots, s_k) \in S_2\).
Therefore, \((s_1, \ldots, s_k) \in S_1 \cap S_2\), and so \(M\) satisfies the conclusion.

Rule~(\ref{proj}) is sound for the following reason.
Let \(M\) be a model, and let \(s_1, \ldots, s_k\) be the values of \(\sigma_1, \ldots, \sigma_k\), respectively, in \(M\).
If \(M\) satisfies the premise, then \((s_1, \ldots, s_k) \in S\).
Therefore \((s_1, \ldots, s_{k-r}) \in S'\), and so \(M\) satisfies the conclusion.

Rule~(\ref{superset}) is sound for the following reason.
Let \(M\) be a model, and let \(s_1, \ldots, s_k\) be the values of \(\sigma_1, \ldots, \sigma_k\), respectively, in \(M\).
If \(M\) satisfies the premise, then \((s_1, \ldots, s_k) \in S\).
Therefore, \((s_1, \ldots, s_k) \in S'\), and so \(M\) satisfies the conclusion.

Rule~(\ref{operators}) is sound for the following reason.
Let \(M\) be a model, and let \(s_1, \ldots, s_k\) be the values of \(\sigma_1, \ldots, \sigma_k\), respectively, in \(M\).
If \(M\) satisfies the premise, then \((s_1, \ldots, s_k) \in S\).
In our real-valued logic, we have that (a) \(f_{\alpha}(s_i, s_j) = s_m\) when \(\sigma_m\) is \(\sigma_i \alphaop \sigma_j\) and \(\alpha\) is a binary connective (such as \(\amp\)), and (b) \(f_{\rho}(s_i) = s_j\) when \(\sigma_j\) is \(\rho\sigma_i\) and \(\rho\) is a unary connective (such as \(\neg\)).
Then \((s_1, \ldots, s_k) \in S'\), and \(M\) satisfies the conclusion.

This completes the proof of soundness.
We now prove completeness.
Assume that \(\Gamma \models \gamma\); we must show that \(\Gamma \infers \gamma\).
We can assume without loss of generality that \(\Gamma\) is nonempty, because if \(\Gamma\) is empty, we replace it by a singleton set containing an instance of our Axiom~(\ref{in01}).

Let \(\Gamma = \{\gamma_1, \ldots, \gamma_n\}\).
For \(1 \leq i \leq n\), assume that \(\gamma_i\) is \((\sigma^i_1, \ldots, \sigma^i_{k_i}; S_i)\), and let \(\Gamma_i = \{\sigma^i_1, \ldots, \sigma^i_{k_i}\}\).
Assume that \(\gamma\) is \((\sigma^0_1, \ldots, \sigma^0_{k_0}; S_0\)), and let \(\Gamma_0 = \{\sigma^0_1, \ldots, \sigma^0_{k_0}\}\).
Let \(G\) be the closure of \(\Gamma_0 \cup \Gamma_1 \cup \cdots \cup \Gamma_n\) under subformulas.

For each \(i\) with \(1 \leq i \leq n\), let \(H_i\) be the set difference \(G \setminus \Gamma_i\).
Let \(r_i = |H_i|\).
Let \(H_i = \{\tau^i_1, \ldots \tau^i_{r_i}\}\).
By applying Rule~(\ref{add}), we prove from \(\gamma_i\) the sentence \((\sigma^i_1, \ldots, \sigma^i_{k_i}, \tau^i_1, \ldots, \tau^i_{r_i}; S_i \times [0, 1]^{r_i})\).
Let \(\psi_i\) be the conclusion of Rule~(\ref{operators}) when the premise is \((\sigma^i_1, \ldots, \sigma^i_{k_i}, \tau^i_1, \ldots, \tau^i_{r_i}; S_i \times [0, 1]^{r_i})\).

Let \(\delta_1, \ldots, \delta_p\) be a fixed ordering of the members of \(G\).
Since the set of components of each \(\psi_i\) is \(G\), we can use Rule~(\ref{perm}) to rewrite \(\psi_i\) as a sentence \((\delta_1, \ldots, \delta_p; T_i)\).
Let us call this sentence \(\phi_i\).

Also, since the only rules used in proving \(\phi_i\) from \(\gamma_i\) are Rules~(\ref{perm}), (\ref{add}), and (\ref{operators}), it follows from Lemma~\ref{lem:equiv} that \(\gamma_i\) and \(\phi_i\) are logically equivalent.

We now make use of the notion of minimization.
Let \(T = T_1 \cap \cdots \cap T_n\).
Define \(\phi\) to be the sentence \((\delta_1, \ldots, \delta_p; T)\).
It follows from Lemma~\ref{lem:min} that each \(\psi_i\) is minimized.
So by Lemma~\ref{lem:moremin}, each \(\phi_i\) is minimized.
By Lemma~\ref{lem:moremin} again, \(\phi\) is minimized.

The sentence \(\phi\) was obtained from the sentences \(\phi_i\) by applying Rule~(\ref{intersect}) \(n-1\) times.
It follows from Lemma~\ref{lem:equiv} that \(\phi\) is equivalent to \(\{\phi_1, \ldots, \phi_n\}\).
Since we also showed that \(\gamma_i\) is logically equivalent to \(\phi_i\) for \(1 \leq i \leq n\), it follows that \(\phi\) is logically equivalent to \(\Gamma\).
Hence, since \(\Gamma \models \gamma\), it follows that \(\{\phi\} \models \gamma\).
It also follows that to prove that \(\Gamma \infers \gamma\), we need only show that there is a proof of \(\gamma\) from \(\phi\).

Recall that \(\gamma\) is \((\sigma^0_1, \ldots, \sigma^0_{k_0}; S_0\)), and \(\phi\) is \((\delta_1, \ldots, \delta_p; T)\).
By applying Rule~(\ref{perm}), we can re-order the components of \(\phi\) so that the components start with \(\sigma^0_1, \ldots, \sigma^0_{k_0}\).
We thereby obtain from \(\phi\) a sentence \((\sigma^0_1, \ldots, \sigma^0_{k_0}, \ldots; T')\), which we denote by \(\phi'\).
By Lemma~\ref{lem:equiv} we know that \(\phi\) and \(\phi'\) are logically equivalent.
So \(\{\phi'\} \models \gamma\).
Since \(\phi\) is minimized, so is \(\phi'\), by Lemma~\ref{lem:moremin}.
By applying Rule~(\ref{proj}), we obtain from \(\phi'\) a sentence \((\sigma^0_1, \ldots, \sigma^0_{k_0}; T'')\), which we denote by \(\phi''\).

We now show that \(T'' \subseteq S_0\).
This is sufficient to complete the proof of completeness, since then we can use Rule~(\ref{superset}) to prove \(\gamma\).
If \(T''\) is empty, we are done.
So assume that \((s_1, \ldots, s_{k_0}) \in T''\); we must show that \((s_1, \ldots, s_{k_0}) \in S_0\).

Since \((s_1, \ldots, s_{k_0}) \in T''\), it follows
that there is an extension \((s_1, \ldots, s_{k_0}, \ldots, s_p)\) in \(T'\).
Since \(\phi'\) is minimized, there is a model \(M\) of \(\phi'\) such that the value of \(\sigma^0_i\) is \(s_i\), for \(1 \leq i \leq k_0\).
Since \(\{\phi'\} \models \gamma\), it follows that \(M\) is a model of \(\gamma\).
By definition of what it means for \(M\) to be a model of \(\gamma\), it follows that \((s_1, \ldots, s_{k_0}) \in S_0\), as desired.

This completes the proof of soundness and completeness.
\end{proof}

\section{Closure of MD-sentences under Boolean combinations}\label{sec:Boolean}

Our next theorem implies that MD-sentences are robust, in that they are closed under Boolean combinations.
Of course, since we are dealing with sentences (which take only the values True and False) in our "outer" logic, we use the standard Boolean connectives.

We begin with a useful lemma that we shall also use later.

\begin{lemma}\label{lem:neg}
    The (standard logical) negation of the sentence \((\sigma_1, \ldots, \sigma_k; S)\) is \((\sigma_1, \ldots, \sigma_k; \tilde{S})\) where \(\tilde{S}\) is the set difference \([0, 1]^k \setminus S\).
\end{lemma}
\begin{proof}
    We need only show that if \(M\) is a model, then \(M \models (\sigma_1, \ldots, \sigma_k; \tilde{S})\) if and only if \(M \not\models (\sigma_1, \ldots, \sigma_k; S)\).
    Let \(s_i\) be the value of \(\sigma_i\) in \(M\), for \(1 \leq i \leq k\).
    If \(M \models (\sigma_1, \ldots, \sigma_k; \tilde{S})\), then \((s_1, \ldots, s_k) \in \tilde{S}\), and so \((s_1, \ldots, s_k) \notin S\).
    Hence, \(M \not\models (\sigma_1, \ldots, \sigma_k; S)\).
    Conversely, if \(M \not\models (\sigma_1, \ldots, \sigma_k; S)\), then \((s_1, \ldots, s_k) \notin S\), and so \((s_1, \ldots, s_k) \in \tilde{S}\).
    Hence, \(M \models (\sigma_1, \ldots, \sigma_k; \tilde{S})\).
\end{proof}

\begin{theorem}\label{thm:robust}
    MD-sentences 
    are closed under Boolean combinations $\land$, $\lor$, and $\neg$.
\end{theorem}
\begin{proof}
Let \(\gamma_1\) and \(\gamma_2\) be MD-sentences.
Assume that \(\gamma_1\) is \((\sigma^1_1, \ldots, \sigma^1_m; S_1)\), and that \(\gamma_2\) is \((\sigma^2_1, \ldots, \sigma^2_n; S_2)\).
As in the proof of Theorem~\ref{thm:soundcomplete}, let \(G\) be the closure of \(\{\sigma^1_1, \ldots, \sigma^1_m, \sigma^2_1, \ldots, \sigma^2_n\}\) under subformulas.
Assume that \(G = \{\delta_1, \ldots, \delta_p\}\).
As in the proof of Theorem~\ref{thm:soundcomplete}, we know that for \(i = 1\) and \(i = 2\), there is \(T_i\) such that \(\gamma_i\) is equivalent to a sentence \((\delta_1, \ldots, \delta_p; T_i)\).
The conjunction \(\gamma_1 \land \gamma_2\) is equivalent to \((\delta_1, \ldots, \delta_p; T_1 \cap T_2)\).
The disjunction \(\gamma_1 \lor \gamma_2\) is equivalent to \((\delta_1, \ldots, \delta_p; T_1 \cup T_2)\).
And by Lemma~\ref{lem:neg}, the negation \(\neg\gamma_1\) is equivalent to \((\delta_1, \ldots, \delta_p; \tilde{T_i})\), where \(\tilde{T_1}\) is the set difference \([0, 1]^p \setminus T_1\).
\end{proof}

\section{Lowering the dimensionality}\label{sec:dimension}

It is natural to ask whether
there is a $(k+1)$-dimensional MD-sentence that in {\L}ukasiewicz or G\"odel logic is not equivalent to any $k$-dimensional MD-sentence.
For the special case $k=1$, the next theorem gives an answer.  We shall shortly state the more general case and generalizations of it as open problems.

\begin{theorem}\label{thm:notequiv}
There is a 2-dimensional MD-sentence that is not equivalent (in either {\L}ukasiewicz or G\"odel logic) to a 1-dimensional MD-sentence.
\end{theorem}
\begin{proof}
 Let $\sigma$ be the 2-dimensional MD-sentence $(A_1,A_2;S)$ where $S = \{(a_1,a_2): a_1^2 = a_2\}$.
 We now show that $\sigma$ is not equivalent to a 1-dimensional MD-sentence.
 If $\phi$ is a propositional formula involving only $A_1$ and $A_2$, then it is easy to see (by induction on the structure of formulas) that for {\L}ukasiewicz or G\"odel logic, $\phi$ defines a piecewise linear function $f_\phi$, in the sense that
 the 1-dimensional MD-sentence $(\phi;S')$ says
 that if $a_1$ is the value of $A_1$ and $a_2$ is the value of $A_2$, then $f_\phi(a_1,a_2) \in S'$.
 Since there is no such piecewise linear function $f_\phi$ and set $S'$  for our sentence $\sigma$, the result holds.
 \end{proof}
 
 The next theorem does not depend on restricting to
 {\L}ukasiewicz or G\"odel logic.
\begin{theorem}\label{thm:onlyk}
Every 
finite set of MD-sentences of arbitrary dimensions 
that involve only the $k$ predicate symbols $A_1, \ldots, A_k$ is equivalent to a single $k$-dimensional MD sentence $(A_1, \ldots, A_k; S)$.  (The set $S$ depends on the real-valued logic being considered.)
\end{theorem}
\begin{proof}
Let $\Gamma$ be a finite set of MD-sentences.
We can view $\Gamma$ as a conjunction of MD-sentences, so by
Theorem~\ref{thm:robust}, $\Gamma$ is equivalent to a single MD-entence $\gamma$.
As in the proof of completeness,  
by closing under subformulas, applying Rule~(\ref{operators}), and reodering by applying Rules~(\ref{perm}),
we obtain an MD-sentnece $(A_1, \ldots, A_k, \phi_1, \ldots, \phi_r';S')$ that is equivalent to $\gamma$.
Since the tuples in $S'$ are good tuples, this is equivalent to the sentence $(A_1, \ldots, A_k;S)$ where $S = \{(s_1, \ldots, s_k) : (s_1, \ldots s_k, s'_1, \ldots s'_r) \in S'\}$.
\end{proof}

{\bf Open problems:}
For each $k$ with $k \geq 2$, does there exist a $(k+1)$-dimensional MD-sentence 
that in {\L}ukasiewicz or G\"odel logic is not equivalent to a $k$-dimensional MD-sentence?
And for $k \geq 1$, how about there being a 
$(k+1)$-dimensional MD-sentence not equivalent to 
a Boolean combination of 1-dimensional MD-sentences, or even to a Boolean combination of $k$-dimensional MD-sentences?

\section{SoCRAtic logic: A decision procedure}\label{sec:decision}

Given a finite set \(\Gamma\) of MD-sentences, and a single MD-sentence \(\gamma\), Theorem~\ref{thm:soundcomplete} says that \(\Gamma \models \gamma\) if and only if \(\Gamma \infers \gamma\).
As we shall show, under natural assumptions there is an algorithm for deciding if \(\Gamma \models \gamma\).
We call this algorithm a \emph{decision procedure}.
If the information sets $S$ all have s simple structure  and the size of $\Gamma$ is treated as a constant, than the algorithm runs in polynomial time.

It is natural to wonder whether we can simply use our complete axiomatization to derive a decision procedure.
The usual answer is that it is not clear in what order to apply the rules of inference.
In our proof of completeness, the rules of inference are applied in a specific order, so that is not an issue here.
Rather, the problem is that in applying Rule~(\ref{operators}), there is no easy way to derive \(S'\) from \(S\), even if \(S\) is fairly simple.
In fact, we now show that even deciding if \(S'\) is nonempty is NP-hard.
Let \(\phi\) be an instance of the NP-hard problem 3SAT.
Thus, \(\phi\) is of the form \((B^1_1 \veebar B^1_2 \veebar B^1_3) \amp \cdots \amp (B^r_1 \veebar B^r_2 \veebar B^r_3)\), where each \(B^i_j\) is a literal (an atomic proposition or its negation).
Assume that the atomic propositions that appear in \(\phi\) are \(A_1, \ldots, A_k\).
Let \(\psi\) be the sentence
\[
    (A_1, \ldots, A_k, \neg A_1, \ldots, \neg A_k, \tau_1, \ldots, \tau_r, \tau_1 \veebar B^1_3, \ldots, \tau_r \veebar B^r_3; S),
\]
where \(\tau_i\) is \(B^i_1 \veebar B^i_2\), for \(1 \leq i \leq r\), and where \(S = \{0, 1\}^{2k+r} \times \{1\}^r\).
Assume that we apply Rule~(\ref{operators}) where the premise is \(\psi\), and the conclusion is
\[
    (A_1, \ldots, A_k, \neg A_1, \ldots, \neg A_k, \tau_1, \ldots, \tau_r, \tau_1 \veebar B^1_3, \ldots, \tau_r \veebar B^r_3; S').
\]
We call this sentence \(\psi'\).
It follows easily from our construction of \(\psi\) that the 3SAT problem \(\phi\) is satisfiable if and only if \(\psi\) is satisfiable.
Now \(\psi\) and \(\psi'\) are logically equivalent, by Lemma~\ref{lem:equiv}.
So the 3SAT problem \(\phi\) is satisfiable if and only if \(\psi'\) is satisfiable.
By Lemma~\ref{lem:min}, we know that \(\psi'\) is minimized.
Hence, if \(S'\) is nonempty, there is a model of \(\psi'\), by the definition of minimization.
And if \(S'\) is empty, then by the definition of a model of a sentence, there is no model of \(\psi'\).
Therefore, \(\psi'\) is satisfiable if and only if \(S'\) is nonempty.
By combining this with our earlier observation that the 3SAT problem \(\phi\) is satisfiable if and only if \(\psi'\) is satisfiable, it follows that the 3SAT problem \(\phi\) is satisfiable if and only if \(S'\) is nonempty.
Hence, deciding if \(S'\) is nonempty is NP-hard.

We now discuss our decision procedure.
To have a chance of there being a decision procedure, the set portion \(S\) of an MD-sentence \((\sigma_1, \ldots, \sigma_k; S)\) must be tractable.
We now give a simple, natural choice for the set portions.
A \emph{rational interval} is a subset of \([0, 1]\) that is of one of the four forms \((a, b)\), \([a, b]\), \((a, b]\), or \([a, b)\), where \(a\) and \(b\) are rational numbers.
Let us say that a sentence \((\sigma_1, \ldots, \sigma_k; S)\) is \emph{interval-based} if \(S\) is of the form \(S_1 \times \cdots \times S_k\), where each \(S_i\) is a union of a finite number of rational intervals.
If each \(S_i\) is the union of at most \(N\) rational intervals, then we say that the sentence is \emph{\(N\)-interval-based}.
Note that this interval-based sentence \((\sigma_1, \ldots, \sigma_k; S)\) is equivalent to the set \(\{(\sigma_1; S_1), \ldots, (\sigma_k; S_k)\}\) of sentences with only one component each.
This observation may be useful in implementing a decision procedure.
In fact, although we do not make use of this in the decision procedure described in this section, we so use it in the implementation of the decision procedure described later, since these sentences with a single component are easy to deal with.   (This is one of  several ways that our implementation differs from what is described in this section.)

Let \(\Gamma = \{\gamma_1, \ldots, \gamma_n\}\).
For \(1 \leq i \leq n\), assume that \(\gamma_i\) is \((\sigma^i_1, \ldots, \sigma^i_{k_i}; S_i)\), and let \(\Gamma_i = \{\sigma^i_1, \ldots, \sigma^i_{k_i}\}\).
Assume that \(\gamma\) is \((\sigma^0_1, \ldots, \sigma^0_{k_0}; S_0\)), and let \(\Gamma_0 = \{\sigma^0_1, \ldots, \sigma^0_{k_0}\}\).
Let \(G\) be the closure of \(\Gamma_0 \cup \Gamma_1 \cup \cdots \cup \Gamma_n\) under subformulas.
If \(|G| \leq M\), then we say that the pair \((\Gamma, \gamma)\) \emph{has nesting depth at most \(M\)}.

\begin{theorem}\label{thm:decision}
Assume either 
 {\L}ukasiewicz logic or G\"odel logic, with the connective \(\amp\), \(\veebar\), \(\rightarrow\), and \(\neg\).\footnote{%
        For {\L}ukasiewicz logic, we could allow each of strong and weak disjunction and conjunction, respectively, as described in an earlier footnote.
    }
Assume that \(\Gamma \cup \{\gamma\}\) is interval based.   Then there is an algorithm that determines whether \(\Gamma \models \gamma\). 
    Assume that \(\Gamma\) has at most \(P\) sentences, each sentence in \(\Gamma \cup \{\gamma\}\) is \(N\)-interval based, and \((\Gamma, \gamma)\) has nesting depth at most \(M\).  If  \(M\) is fixed, then the algorithm runs in time polynomial in \(P\) and \(N\).

\end{theorem}
\begin{proof}
Assume throughout the proof that \(\Gamma\) has at most \(P\) sentences, each sentence in \(\Gamma \cup \{\gamma\}\) is \(N\)-interval based, and \((\Gamma, \gamma)\) has nesting depth at most \(M\).

It is easy to see that \(\Gamma \models \gamma\) if and only \(\Gamma \cup \{\neg\gamma\}\) is not satisfiable.
So we need only give an algorithm that decides whether \(\Gamma \cup \{\neg\gamma\}\) is satisfiable.

Let \(\{\sigma_1, \ldots, \sigma_p\}\) be the closure of \(\Gamma \cup \{\gamma\}\) under subformulas.
Let \(\Gamma = \{\gamma_1, \ldots, \gamma_n\}\).
By making use of Rules~(\ref{perm}) and (\ref{add}), for each \(i\) with \(1 \leq i \leq n\), we can create a sentence \(\gamma'_i\) of the form \((\sigma_1, \ldots, \sigma_p; S^i)\) that by Lemma~\ref{lem:equiv}  is equivalent to \(\gamma_i\), and that has \(\sigma_1, \ldots, \sigma_p\) as components.
By the construction, each \(\gamma_i'\) is \(N\)-interval-based.

Similarly, create the sentence \(\gamma'\) of the form \((\sigma_1, \ldots, \sigma_p; T)\) that is equivalent to \(\gamma\), and that has \(\sigma_1, \ldots, \sigma_p\) as components.
As before, \(\gamma'\) is \(N\)-interval-based.

Now \(\Gamma\) is equivalent to the conjunction of the sentences \(\gamma'_i\) for \(1 \leq i \leq n\), and this conjunction is equivalent to \((\sigma_1, \ldots, \sigma_p; S)\), where \(S = \bigcap_{i \leq n} S^i\).
We now show that \((\sigma_1, \ldots, \sigma_p; S\)) is \(PN\)-interval-based.
By assumption, for each \(i\) with \(1 \leq i \leq n\), we have that \(S^i\) is of the form \(S^i_1 \times \cdots \times S^i_p\), where each \(S^i_j\) is the union of at most \(N\) intervals.
For each \(j\) with \(1 \leq j \leq p\), let \(S_j = \bigcap_i S^i_j\).
Then \(S = S_1 \times \cdots \times S_p\).
So to show that \((\sigma_1, \ldots, \sigma_p; S\)) is \(PN\)-interval-based, we need only show that each \(S_j\) is the union of at most \(PN\) intervals.

Since \(S_j = \bigcap_{i \leq n} S^i_j\), where each \(S^i_j\) is the union of at most \(N\) intervals, we see that \(S_j\) is the union of intervals where the left endpoint of each interval in \(S_j\) is one of the left endpoints of intervals in \(\bigcup_{i \leq n} S^i_j\).
For each \(j\), there are \(n\) sets \(S^i_j\).
And for each \(i\) with \(1 \leq i \leq n\), there are at most \(N\) left endpoints of \(S^i_j\).
So the total number of left endpoints of intervals in \(\bigcup_{i \leq n} S^i_j\) is at most \(nN \leq PN\), and so the number of intervals in \(S_j\) is at most \(PN\).
Since \(S = S_1 \times \cdots \times S_p\), it follows that \((\sigma_1, \ldots, \sigma_p; S)\) is \(PN\)-interval-based.

Let us now consider \(\neg\gamma\), which is equivalent to \(\neg\gamma'\).
Recall that \(\gamma'\) is \((\sigma_1, \ldots, \sigma_p; T)\), and that \(\gamma'\) is \(N\)-interval-based.
So \(T\) is of the form \(T_1 \times \cdot \times T_p\), where each \(T_j\) is the union of at most \(N\) intervals.
By Lemma~\ref{lem:neg}, the negation of \(\gamma'\) is \((\sigma_1, \ldots, \sigma_p; \tilde{T})\), where \(\tilde{T}\) is the set difference \([0, 1]^p \setminus T\).
For each \(j\) with \(1 \leq j \leq p\), let \(T'_j\) be the set difference \([0, 1] \setminus T_j\).
Clearly, \(T'_j\) is the union of intervals.
The left endpoints of intervals in \(T'_j\) are the right-end points of intervals in \(T_j\), possible along with 0.
So \(T'_j\) is the union of at most \(N + 1\) intervals.
Let \(V_j = [0, 1]^{j-1} \times T'_j \times [0, 1]^{p-j}\).
It is straightforward to see that \(\tilde{T} = \bigcup_{j \leq p} V_j\).

Now, showing that \(\Gamma \cup \{\neg\gamma\}\) is not satisfiable is equivalent to showing that \((\sigma_1, \ldots, \sigma_p; S) \land (\sigma_1, \ldots, \sigma_p; \tilde{T})\) is not satisfiable, which is equivalent to showing that for every \(j\) with \(1 \leq j \leq p\), we have that \((\sigma_1, \ldots, \sigma_p; S) \land (\sigma_1, \ldots, \sigma_p; V_j)\) is not satisfiable.
So we need only give an algorithm for deciding if \((\sigma_1, \ldots, \sigma_p; S) \land (\sigma_1, \ldots, \sigma_p; V_j)\) is satisfiable.
Let us hold \(j\) fixed.
Since, as we showed, \((\sigma_1, \ldots, \sigma_p; S)\) is \(PN\)-interval-based, we can write \(S\) as \(S_1 \times \cdots \times S_p\), where each \(S_i\) is the union of at most \(PN\) intervals.
Now \((\sigma_1, \ldots, \sigma_p; S) \land (\sigma_1, \ldots, \sigma_p; V_j)\) is equivalent to \((\sigma_1, \ldots, \sigma_p; S \cap V_j)\).
Now \(S \cap V_j\) is of the form \(S'_1 \times \cdots \times S'_p\), where \(S'_m = S_m\) for \(m \neq j\), and where \(S'_j = S_j \cap T'_j\).
We showed that \(T'_j\) is the union of at most \(N + 1\) intervals, and that \(S_j\) is the union of at most \(PN\) intervals, so it follows that \(S_j \cap T'_j\) is the union of at most \(PN + N + 1\) intervals, since each left endpoint of the intervals in \(S_j \cap T'_j\) is a left endpoint of an interval in \(S_j\) or an interval in \(T'_j\).

We now describe our algorithm for deciding if the sentence \((\sigma_1, \ldots, \sigma_p; S \cap V_j)\), that is, for the sentence \((\sigma_1, \ldots, \sigma_p; S'_1 \times \cdots \times S'_p\)), which is \((PN+N+1)\)-interval-based, is satisfiable.
This can be broken into \(|S'_1| \times \cdots \times |S'_p|\) subproblems, one for each choice \((I_1, \ldots, I_p)\) of a single interval \(I_k\) from \(S'_k\) for each \(k\) with \(1 \leq k \leq p\).
This gives a total of at most \((PN+N+1)^M\) subproblems.
For each of these subproblems, we wish to decide satisfiability of the system \(\{s_1 \in I_1, \ldots, s_p \in I_p\}\) along with (a) the binary constraints \(f_{\alpha}(s_i, s_j) = s_m\) when \(\sigma_m\) is \(\sigma_i \alphaop \sigma_j\) and \(\alpha\) is a \(\amp\), \(\veebar\), or \(\Rightarrow\), and (b) \(f_{\neg}(s_i) = s_j\) when \(\sigma_j\) is \(\neg\sigma_i\).

The constraints \(s_j \in I_j\) are specified by inequalities (for example, if \(I_j\) is \((a, b]\) we get the inequalities \(a < s_i \leq b\)).
We now show how to deal with the constraints in (a) and (b) above.
A canonical example is given by dealing with \(f_{\amp}(s_i, s_j) = s_m\) in G\"odel logic, which interprets ``\(f_{\amp}(s_i, s_j) = s_m\)'' as \(\min \{s_i, s_j \} = s_m\).
We split the system of constraints into two systems of constraints, one where we replace \(\min\{s_i, s_j \} = s_m\) by the two statements ``\(s_i \leq s_j\), \(s_i = s_m\)'' and another where we replace \(\min\{s_i, s_j \} = s_m\) by the two statements ``\(s_j < s_i\), \(s_j = s_m\)''.
In {\L}ukasiewicz logic, where \(f_{\amp}(s_i, s_j)\) is \(\max\{0, s_1 + s_2 - 1\}\), we split the system of constraints into two systems of constraints, one where we replace \(\max\{0, s_1 + s_2 - 1\} = s_m\) by the two statements ``\(s_i + s_j - 1 \geq 0\), \(s_i + s_j - 1 = s_m\)'' and another where we replace \(\max\{0, s_1 + s_2 - 1\} = s_m\) by the two statements ``\(s_i + s_j - 1 < 0\), \(s_m = 0\)''.
The same approach works for the other binary connectives.
For example, in G\"odel logic, where \(f_{\Rightarrow}(s_i, s_j)\) is 1 if \(s_i \leq s_j\) and is \(s_j\) otherwise, we would split into two case, one where we replace \(f_{\Rightarrow}(s_i, s_j) = s_m\) by the two statements ``\(s_i \leq s_j\), \(s_m = 1\)'' and another where we replace \(f_{\Rightarrow}(s_i, s_j) = s_m\) by the two statements ``\(s_j > s_i\), \(s_m = s_j\)''.
In considering the effect of the constraints in (a) and (b), each of our at most \((PN+N+1)^M\) subproblems splits at most \(2^p \leq 2^M\) times, giving a grand total of at most \((PN+N+1)^M 2^M\) systems of inequalities that we need to check for feasibility (that is, to see if there is a solution).
For each of these systems of inequalities, we can make use a polynomial-time algorithm for linear programming
to decide feasibility,
where the size of each of these systems is linear in $M$, and so the running time for  each instance of the linear programming algorithm is polynomial in $M$.
Since also the number of systems is at most \((PN+N+1)^M 2^M\), and since\(M\) is fixed  by assumption, this gives us an overall algorithm for decidability, whose rulnning time is polynoimial in $N$ and $P$.
\end{proof}

The reason we held the parameter \(M\) fixed is that the running time of the algorithm is exponential in $M$, because there are an exponential number of calls to a linear programming subroutine.
The algorithm is polynomial-time if there is a fixed bound on $M$.
Such a bound is necessary, because the problem 
can be co-NP hard, for the following reason.

Let \(\gamma\) be the sentence \((A, \neg A; [1] \times [1])\).
Then \(\gamma\) is not satisfiable.
Let \(\Gamma\) consist of the single sentence \(\psi\) from the beginning of the section.
Then \(\Gamma \models \gamma\) if and only if \(\psi\) is not satisfiable.
Now \(\psi\) is satisfiable if and only if \(S'\) from the beginning of the section is nonempty, which we showed is an NP-hard problem to determine.
Since \(\Gamma \models \gamma\) if and only if \(\psi\) is not satisfiable, it follows that deciding if \(\Gamma \models \gamma\) is co-NP hard.

We now give an implementation of the decisoin procedure
The decision procedure described in Section~\ref{sec:decision} is available under the \texttt{socratic-logic} subdirectory provided with this supplementary material.
We implemented the algorithm as a Python package named \texttt{socratic}, which requires Python 3.6 or 3.7 and makes use of IBM\textsuperscript{\textregistered} ILOG\textsuperscript{\textregistered} CPLEX\textsuperscript{\textregistered} Optimization Studio V12.10.0 via the \texttt{docplex} Python package.

\subsection{Source code organization}

The source code is organized as follows:
\begin{description}[font=\ttfamily,leftmargin=!,labelwidth=1.25in]
    \item[setup.sh]           A script to create a Python \texttt{virtualenv} and install required packages
    \item[requirements.txt]   A standard \texttt{pip} list of package dependencies
    \item[socratic/theory.py] Implementations for theories, sentences, and bounded intervals
    \item[socratic/op.py]     Implementations for each logical operator as well as propositions and truth value constants
    \item[socratic/demo.py]   Two example use cases demonstrating the use of the package
    \item[socratic/test.py]   A suite of unit tests also serving as our experimental setup
    \item[socratic/hajek.py]  Many tautologies proved in \cite{hajek1998metamathematics} used in \texttt{test.py}
    \item[socratic/clock.py]  A higher-order function to measure the runtime of experiments
\end{description}
The classes defined in the source code are:
\begin{description}[font=\ttfamily,leftmargin=.5\leftmargin]
    \item[theory.Theory]\ \\A collection of sentences that can test for satisfiability or the entailment of a query sentence under a given logic
    \item[theory.Sentence]\ \\A base-class for a collection of formulas and an associated set of candidate interpretations for the formulas
    \begin{description}[font=\ttfamily,leftmargin=.5\leftmargin]
        \item[theory.SimpleSentence]\ \\A single formula and an associated collection of candidate truth value intervals for the formula
    \end{description}
    \item[theory.FloatInterval]\ \\An open or closed interval of truth values
    \begin{description}[font=\ttfamily,leftmargin=.5\leftmargin]
        \item[theory.ClosedInterval]\ \\\([l, u]\), i.e., all values from $l$ to $u$, inclusive
        \begin{description}[font=\ttfamily,leftmargin=.5\leftmargin]
            \item[theory.Point]\ \\\([x, x]\), i.e., just $x$
        \end{description}
        \item[theory.OpenInterval]\ \\\((l, u)\), i.e., all values from $l$ to $u$, exclusive
        \item[theory.OpenLowerInterval]\ \\\((l, u]\), i.e., all values from $l$ to $u$ excluding $l$
        \item[theory.OpenUpperInterval]\ \\\([l, u)\), i.e., all values from $l$ to $u$ excluding $u$
    \end{description}
    \item[op.Formula]\ \\The base-class of a data structure representing the syntax tree of a logical formula
    \begin{description}[font=\ttfamily,leftmargin=.5\leftmargin,itemindent=!,labelwidth=.75in]
        \item[op.Prop] A named proposition, e.g., \(x\)
        \item[op.Constant] A truth value constant, e.g., \(.5\)
        \item[op.Operator] A base-class for all formulas with subformulas
        \begin{description}[font=\ttfamily,leftmargin=!,labelwidth=.75in]
            \item[op.And] Strong conjunction \(x \otimes y\)
            \item[op.WeakAnd] Weak conjunction \(x \amp y = \min\{x, y\}\)
            \item[op.Or] Strong disjunction \(x \otimes y\)
            \item[op.WeakOr] Weak disjunction \(x \veebar y = \max\{x, y\}\)
            \item[op.Implies] Implication \(x \Rightarrow y\), i.e., the residuum of \(\otimes\)
            \item[op.Not] Negation defined \(\neg x = (x \Rightarrow 0)\)
            \item[op.Inv] Involute negation \({\sim}x = 1 - x\)
            \item[op.Equiv] Logical equivalence defined \((x \equiv y) = ((x \Rightarrow y) \otimes (y \Rightarrow x))\)
            \item[op.Delta] The operation \(\Delta x = 1\) if \(x = 1\) else 0
        \end{description}
    \end{description}
\end{description}

\subsection{Implementation details}

The implementation strategy closely adheres to the decision procedure described in Section~\ref{sec:decision}, though with a few notable design shortcuts.

\paragraph{Boolean variables.}
One such shortcut is the use of mixed integer linear programming (MILP) to perform the ``spliting'' of linear programs into two possible optimization problems, specifically by adding a Boolean variable that determines which of a set of constraints must be active.
For example, given the desired constraint \(z = \min\{x, y\}\), one may write
\begin{align}
    z & \leq x \\
    z & \leq y \\
    z & \geq x - (1 - b) \label{eqn:z-eq-x} \\
    y & \geq x - (1 - b) \label{eqn:y-ge-x} \\
    z & \geq y - b \label{eqn:z-eq-y} \\
    x & \geq y - b \label{eqn:x-ge-y}
\end{align}
for Boolean variable \(b\).
For \(x, y, z \in [0, 1]\), observe that (\ref{eqn:z-eq-x}) and (\ref{eqn:y-ge-x}) are effectively disabled for \(b = 0\) and that (\ref{eqn:z-eq-y}) and (\ref{eqn:x-ge-y}) are likewise disabled for \(b = 1\).
For example, when \(b = 1\), the remaining constraints are \(z \leq x, z \leq y, z \geq x, y \geq x\), which is equivalent to \(z = x, x \leq y\), as desired.
Observe then that MILP's exploration of either value for the Boolean variable is equivalent to repeating linear optimization for either possible set of constraints; no feasible solution exists for any combination of such Boolean variables in exactly the case that none of the split linear programs are feasible.
In practice, CPLEX has built-in support for min, max, abs, and a handful of other useful functions, though the above technique is still required to implement G\"odel logic's residuum, negation, and equivalence operations as well as to select the specific intervals a sentence's formula truth values lie within.

\paragraph{Strict inequality.}
The described decision procedure also occasionally calls for continuous constraints with strict inequality, in particular when dealing with the complements of closed intervals, but also when handling input open intervals or the G\"odel residuum, \((x \Rightarrow y) = y\) if \(x > y\) else 1.
Linear programming, however, does not inherently support this.
To implement strict inequality constraints, we introduce a global gap variable \(\delta \in [0, 1]\) to widen the distance between either side of the inequality, e.g.
\begin{equation}
    x \geq y + \delta, \label{eqn:gap}
\end{equation}
and then seek to maximize \(\delta\).
If optimization yields an apparently feasible solution but with \(\delta = 0\), we regard it as infeasible because at least one strict inequality constraint could not be honored strictly.
Again in practice, due to floating-point imprecision, MILP can sometimes return tiny though nonzero values of \(\delta\) even for \(x = y\) in (\ref{eqn:gap}); as a result, it is necessary to check if \(\delta\) is greater than some threshold rather than merely nonzero.
We use \(\delta > 10^{-8}\), which is much larger than the imprecision we have observed and yet much smaller than most truth values we consider.
We observe that this technique is roughly equivalent to replacing \(\delta\) with \(10^{-8}\) throughout, which has the added benefit of freeing up the optimization objective for other uses in future extensions of the decision procedure, such as determining the tightest bounds for which a theory can entail a query.

\paragraph{Simple sentences.}
We additionally observe that, for theories restricted to interval-based sentences, it is sufficient to support only sentences containing a single formula and collection of truth value intervals, i.e., \texttt{SimpleSentences} of the form \((\sigma; S)\) for a single formula~\(\sigma\).
This is because of the following theorem:

\begin{theorem}
    Any interval-based sentence~\(s = (\sigma_1, \ldots, \sigma_k; S_1 \times \cdots \times S_k)\) is equivalent to a collection of simple sentences~\(s_1, \ldots, s_k\), each given \(s_i = (\sigma_i; S_i)\).
\end{theorem}
\begin{proof}
    Given interval-based sentence~\(s\) and simple sentences~\(s_1, \ldots, s_k\) as described, one may apply Rules~(\ref{add}) and (\ref{perm}) to obtain \(s_1', \ldots, s_k'\) given \(s_i' = (\sigma_1, \ldots, \sigma_k; [0, 1]^{i-1} \times S_i \times [0, 1]^{k-i})\).
    One may then repeatedly apply Rule~(\ref{intersect}) to compose these exactly into \(s\).
    Likewise, one may apply Rules~(\ref{perm}) and (\ref{proj}) to obtain each \(s_i\) directly from \(s\).
    Hence, the two forms are equivalent.
\end{proof}

Accordingly, \texttt{socratic} implements only \texttt{SimpleSentence}.
In order to include an interval-based sentence in a theory, one should instead include each of its component simple sentences as constructed above.
In order to test the entailment of an interval-based sentence, one should separately test the entailment of each of its component simple sentences.

\paragraph{Complementary intervals.}
As a last deviation from the described decision procedure, rather than explicitly finding the complement of a collection of truth value intervals for a given query formula, we simply adjust how constraints are expressed to force feasible solutions into the set of complementary intervals.
Specifically, while the usual constraints require the formula's truth value to lie within some one of its intervals, the complementary constraints require the formula's truth value to not lie in any of its intervals, i.e., to lie to the left or the right of each of its intervals.
We then reverse the direction of each interval's lower and upper bound constraints, adding or removing gap variable~\(\delta\) as appropriate to switch between strict and nonstrict inequalities, and introduce Boolean parameters to decide which of each pair of constraints should apply, i.e., to decide whether the formula's truth value should lie to the left or to the right.
For example, if simple sentence has truth value intervals \([.2, .3], (.5, 1]\), the above would produce constraints
\begin{align*}
    x & \leq .2 - \delta + (1 - b_1), \\
    x & \geq .3 + \delta - b_1, \\
    x & \leq .5 + (1 - b_2), \\
    x & \geq 1 + \delta - b_2
\end{align*}
Observe that the last of these cannot be satisfied for \(\delta > 0\) unless \(b_2 = 1\), which is consistent with the complement of \((.5, 1]\) not having a right side in the interval~\([0, 1]\).

\subsection{Experimental results}

We tested \texttt{socratic} in three different experimental contexts:
\begin{itemize}
    \item 3SAT and higher \(k\)-SAT problems which become satisfiable if any one of their input clauses is removed
    \item 82 axioms and tautologies taken from H\'ajek in \cite{hajek1998metamathematics}, some of which hold only for one of {\L}ukasiewicz or G\"odel logic
    \item A formula that is classically valid but invalid in both {\L}ukasiewicz and G\"odel logic, unless propositions are constrained to be Boolean
    \item A stress test running \texttt{socratic} on sentences with thousands of intervals
\end{itemize}
All experiments are conducted on a MacBook Pro with a 2.9 GHz Quad-Core Intel Core i7, 16 GB 2133 MHz LPDDR3, and Intel HD Graphics 630 1536 MB running macOS Catalina 10.15.5.

\paragraph{\(k\)-SAT.}
We construct classically unsatisfiable \(k\)-SAT problems of the form
\begin{equation}
    (x_1 \amp \neg x_1) \veebar \cdots \veebar (x_k \amp \neg x_k)
\end{equation}
which, after CNF conversion, yields for 3SAT
\begin{align*}
    & (x_1 \veebar x_2 \veebar x_3), &
    & (\neg x_1 \veebar x_2 \veebar x_3), &
    & (x_1 \veebar \neg x_2 \veebar x_3), \\
    & (x_1 \veebar x_2 \veebar \neg x_3), &
    & (x_1 \veebar \neg x_2 \veebar \neg x_3), &
    & (\neg x_1 \veebar x_2 \veebar \neg x_3), \\
    & (\neg x_1 \veebar \neg x_2 \veebar x_3), &
    & (\neg x_1 \veebar \neg x_2 \veebar \neg x_3)
\end{align*}
and similarly for larger \(k\).
The removal of any one clause in such a problem renders it satisfiable.
We observe that, when each clause is required to have truth value exactly 1 but propositions are allowed to have any truth value, \texttt{socratic} correctly determines the problem to be
\begin{enumerate}[label=\arabic*)]
    \item unsatisfiable in G\"odel logic,
    \item satisfiable in G\"odel logic when dropping any one clause,
    \item trivially satisfiable in {\L}ukasiewicz logic with, e.g., \(x_i = .5\),
    \item again unsatisfiable in {\L}ukasiewicz logic when propositions are required to have truth values in range either \(\left[0, \frac{1}{k}\right)\) or \(\left(\frac{k - 1}{k}, 1\right]\),
    \item and yet again satisfiable in {\L}ukasiewicz logic with constrained propositions when dropping any one clause.
\end{enumerate}
We observe that G\"odel logic is much slower than {\L}ukasiewicz logic as implemented in \texttt{socratic}, likely because it performs mins and maxes between many arguments throughout while {\L}ukasiewicz logic instead performs sums with simpler mins and maxes serving as clamps to the \([0, 1]\) range.
Interestingly, the difference between unsatisfiable and satisfiable in G\"odel logic is significant; while the satisfiable problems have one fewer clause, this is more likely explained by \texttt{socratic} finding a feasible solution quickly.
On the other hand, the unsatisfiable and satisfiable problems (with constrained propositions) take roughly the same amount of time for {\L}ukasiewicz, though the trivially satisfiable problem is quicker.
The apparent exponential increase in runtime is partially explained by the fact that each larger problem has twice as many clauses, but runtime appears to be growing by slightly more than a factor of 2 per each \(k\).

\begin{table}[tbh]
    \centering
    \caption{%
        \(k\)-SAT runtimes in seconds for \texttt{socratic} with different experimental configurations.
        The five columns pertain to items 1 through 5 above.
        The problem is unsatisfiable in classical and G\"odel logic, satisfiable in G\"odel logic after removing a clause at random, trivially satisfiable in {\L}ukasiewicz logic with, e.g., \(x_i = .5\), unsatisfiable in {\L}ukasiewicz logic if propositions are required to lie in ranges sufficiently close to 0 and 1, and again satisfiable in {\L}ukasiewicz logic with constrained propositions when removing a clause at random.
    }
    \begin{tabular}{r|rrrrr}
    \toprule
            & G\"odel           & G\"odel          & {\L}uka.   & {\L}uka.   & {\L}uka. \\
        \(k\) & unsat.     & satisf.      & trivial          & unsat.    & satisf. \\
    \midrule
          3 &    .012           &   .011           &   .014           &   .019           &   .014 \\
          4 &    .022           &   .020           &   .022           &   .031           &   .033 \\
          5 &    .054           &   .043           &   .041           &   .047           &   .043 \\
          6 &    .121           &   .107           &   .064           &   .104           &   .098 \\
          7 &    .204           &   .255           &   .173           &   .167           &   .206 \\
          8 &    .404           &   .414           &   .273           &   .286           &   .308 \\
          9 &    .861           &   .881           &   .507           &   .539           &   .554 \\
         10 &   5.46\phantom{0} &  1.99\phantom{0} &  1.03\phantom{0} &  1.11\phantom{0} &  1.17\phantom{0} \\
         11 &  18.0\phantom{00} &  4.34\phantom{0} &  2.09\phantom{0} &  2.44\phantom{0} &  2.21\phantom{0} \\
         12 &  33.3\phantom{00} & 10.9\phantom{00} &  4.36\phantom{0} &  5.06\phantom{0} &  5.01\phantom{0} \\
         13 & 119\phantom{.000} & 25.8\phantom{00} &  8.72\phantom{0} & 12.4\phantom{00} & 12.3\phantom{00} \\
         14 & 696\phantom{.000} & 71.0\phantom{00} & 18.4\phantom{00} & 38.0\phantom{00} & 35.6\phantom{00} \\
    \bottomrule
    \end{tabular}
    \label{tab:k-sat}
\end{table}

\paragraph{H\'ajek tautologies.}
H\'ajek lists many axioms and tautologies pertaining to a system of logic he describes as basic logic (BL), consistent with any t-norm logic, as well as a number of tautologies specific to {\L}ukasiewicz and G\"odel logic, all of which should have truth value exactly 1.
We implement these tautologies in \texttt{socratic} and test whether the empty theory can entail them with truth value 1 in their respective logics.
The BL tautologies are divided into batches pertaining to specific operations and properties:
\begin{description}[font=\ttfamily,leftmargin=!,labelwidth=1.125in]
    \item[axioms] 8 tests, e.g., \((\phi \Rightarrow \psi) \Rightarrow ((\psi \Rightarrow \chi) \Rightarrow (\phi \Rightarrow \chi))\)
    \item[implication] 3 tests, e.g., \(\phi \Rightarrow (\psi \Rightarrow \phi)\)
    \item[conjunction] 6 tests, e.g., \((\phi \otimes (\phi \Rightarrow \psi)) \Rightarrow \psi\)
    \item[weak\char`_conjunction] 7 tests, e.g., \((\phi \amp \psi) \Rightarrow \phi\)
    \item[weak\char`_disjunction] 7 tests, e.g., \(\phi \Rightarrow (\phi \veebar \psi)\)
    \item[negation] 8 tests, e.g., \(\phi \Rightarrow (\neg\phi \Rightarrow \psi)\)
    \item[associativity] 6 tests, e.g., \((\phi \amp (\psi \amp \chi)) \Rightarrow ((\phi \amp \psi) \amp \chi)\)
    \item[equivalence] 9 tests, e.g., \(((\phi \equiv \psi) \otimes (\psi \equiv \chi)) \Rightarrow (\phi \equiv \chi)\)
    \item[distributivity] 8 tests, e.g., \((\phi \otimes (\psi \veebar \chi)) \equiv ((\phi \otimes \psi) \veebar (\phi \otimes \chi))\)
    \item[delta\char`_operator] 3 tests, e.g., \(\Delta \phi \equiv \Delta (\phi \otimes \phi)\)
\end{description}
In addition, there are logic-specific batches of tautologies:
\begin{description}[font=\ttfamily,leftmargin=!,labelwidth=1.125in]
    \item[lukasiewicz] 12 tests, e.g., \(\neg\neg\phi \equiv \phi\)
    \item[godel] 5 tests, e.g., \(\phi \Rightarrow (\phi \otimes \phi)\)
\end{description}
Each of the above BL batches complete successfully for both logics and each of the logic-specific batches complete for their respective logics and, as expected, fail for the other logic.
The runtime of individual tests are negligible; the entire test suite of 82 tautologies run on both logics complets in just 2.911 seconds.

\paragraph{Boolean logic.}
We consider a formula \(\sigma\) defined
\begin{equation}
    (\phi \Rightarrow \psi) \Rightarrow ((\neg\phi \Rightarrow \psi) \Rightarrow \psi)
\end{equation}
which is valid in classical logic but is not entailed with truth value 1 by the empty theory in either {\L}ukasiewicz or G\"odel logic.
Conversely, constraining propositions~\(\phi\) and \(\psi\) to have classical truth values by introducing the sentences
\begin{align*}
    & (\phi; [0, 0] \cup [1, 1]), \\
    & (\psi; [0, 0] \cup [1, 1])
\end{align*}
into the theory succeeds in entailing \(\sigma\) in either logic.
Indeed, if even one of these sentences is added, \(\sigma\) is entailed, but no looser intervals around 0 and 1 can entail \(\sigma\) if both propositions are non-Boolean.
In the other direction, {\L}ukasiewicz logic with unconstrained propositions entails the sentence \((\sigma; [.5, 1])\), i.e., \(\sigma\) with truth value bounded above .5, while G\"odel logic with unconstrained propositions cannot entail \(\sigma\) with any interval tighter than \([0, 1]\).
As a final example of the interaction between truth value intervals, G\"odel logic entails \((\sigma; [t, 1])\) for a lower bound truth value~\(t\) if either of \(\phi\) or \(\psi\) is constrained in the theory to have set of candidate truth values \(\{0\} \cup [t, 1]\).

\paragraph{Stress test.}
We consider the experimental configuration for Boolean logic above now with query \((\sigma; S)\) for \(S\) consisting of 10000 open intervals \((\frac{1}{k+1}, \frac{1}{k})\) for \(k\) from 2 to 10000 plus the closed interval \([.5, 1]\) and \((\phi; S')\) and \((\psi; S')\) for \(S'\) consisting of 10000 open intervals \((1 - \frac{1}{k}, 1 - \frac{1}{k+1})\) plus the closed interval \([0, 0]\).
We observe the runtime of \texttt{socratic} to be just 11.8 seconds for G\"odel logic and 9.38 seconds for {\L}ukasiewicz logic.
If we instead use closed intervals throughout, measured runtimes are 17.4 seconds for G\"odel and 9.29 seconds for {\L}ukasiewicz.

\section{Dealing with weights}\label{sec:weights}

In some circumstances, such as logical neural networks \cite{riegel2020logical}, weights are assigned to subformulas, where the weight is intended to reflect the influence, or importance, of the subformula.
Each weight is a real number.
For example, in the formulas \(\sigma_1 \veebar \sigma_2\), the weight \(w_1\) might be assigned to \(\sigma_1\) and the weight \(w_2\) assigned to \(\sigma_2\).
If \(0 < w_1 = 2 w_2\), this might indicate that \(\sigma_1\) is twice as important as \(\sigma_2\) in evaluating the value of \(\sigma_1 \veebar \sigma_2\).

As an example of a possible way to incorporate weights, assume that we are using {\L}ukasiewicz real-valued logic, where the value of \(\sigma_1 \veebar \sigma_2\) is \(\min\{1, s_1 + s_2\}\), when \(s_1\) is the value of \(\sigma_1\) and \(s_2\) is the value of \(s_2\).
If the weights of \(\sigma_1\) and \(\sigma_2\) are \(w_1\) and \(w_2\), respectively, and if both \(w_1\) and \(w_2\) are non-negative, then we might take the value of \(\sigma_1 \veebar \sigma_2\) in the presence of these weights to be \(\min\{1, w_1 s_1 + w_2 s_2\}\).

We now show how to incorporate weights into our approach.
In fact, the ease of incorporating weights and still getting a sound and complete axiomatization is a real advantage of our approach!

To deal with weights, we define an expanded view of what a formula is, defined recursively.
Each atomic proposition is a formula.
If \(\sigma_1\) and \(\sigma_2\) are formulas, \(w_1\) and \(w_2\) are weights, and \(\alpha\) is a binary connective (such as \(\amp\)) then \((\sigma_1 \alphaop \sigma_2, w_1, w_2)\) is a formula.
Here \(w_1\) is interpreted as the weight of \(\sigma_1\) and \(w_2\) as the weight of \(\sigma_2\) in the formula \(\sigma_1 \alphaop \sigma_2\).
Also, if \(\sigma\) is a formula, \(w\) is a weight, and \(\rho\) is a unary connective (such as \(\neg\)), then \((\rho\sigma, w)\) is a formula, where \(w\) is interpreted as the weight of \(\sigma\).
We modify our definition of \emph{subformula} as follows.
The subformulas of \((\sigma_1 \alphaop \sigma_2, w_1, w_2)\) are \(\sigma_1\) and \(\sigma_2\), and the subformula of \((\rho\sigma, w)\) is \(\sigma\). %

If \(\alpha\) is a binary connective, then \(f_{\alpha}\) now has four arguments, rather than two.
Thus, \(f_{\alpha}(s_1, s_2, w_, w_2)\) is the value of the formula \((\sigma_1 \alphaop \sigma_2, w_1, w_2)\) when the value of \(\sigma_1\) is \(s_1\), the value of \(\sigma_2\) is \(s_2\), the weight of \(\sigma_1\) is \(w_1\), and the weight of \(\sigma_2\) in \(w_2\).
Also, \(f_{\rho}\) now has two arguments rather than one.
Thus, \(f_{\rho}(s, w)\) is the value of the formula \((\rho\sigma, w)\) when the value of \(\sigma\) is \(s\), and the weight of \(\sigma\) is \(w\).

The axiom and rules are just as before, except that Rule~(\ref{operators}) is changed to:
\begin{equation}\label{operators1}
    \mbox{From } (\sigma_1, \ldots, \sigma_k; S) 
    \mbox{ infer } (\sigma_1, \ldots, \sigma_k; S')
\end{equation}
where \(S' = \{(s_1, \ldots, s_k)\colon (s_1, \ldots, s_k, \in S\) and (a) \(s_m = f_{\alpha}(s_i, s_j, w_1, w_2)\) when \(\sigma_m\) is \((\sigma_i \alphaop \sigma_j, w_1, w_2)\) and \(\alpha\) is a binary connective, and (b) \(s_j = f_{\rho}(s_i, w) \) when \(\sigma_j\) is \((\rho\sigma_i, w)\) and \(\rho\) is a unary connective\(\}\).

We can extend Theorem~\ref{thm:soundcomplete} (soundness and completeness) and Theorem~\ref{thm:robust} (closure under Boolean combinations) to deal with our sentences \((\sigma_1, \ldots, \sigma_k; S)\) that include weights.
The proofs go through just as before, where we use Rule~(\ref{operators1}) instead of Rule~(\ref{operators}).
Thus, we obtain the following theorems.

\begin{theorem}\label{thm:soundcompleteweights}
    Our axiom system where Rule~(\ref{operators}) is replaced by Rule~(\ref{operators1}) is sound and complete for sentences of the form \((\sigma_1, \ldots, \sigma_k; S)\) that include weights.
\end{theorem}

\begin{theorem}\label{thm:robust1}
    The sentences \((\sigma_1, \ldots, \sigma_k; S)\) that include weights are closed under Boolean combinations.
\end{theorem}

What about the decision procedure in Theorem~\ref{thm:decision}?
The key step 
is the use of a polynomial-time algorithm for linear programming.
If we were to hold the weights \(w_i\) as fixed rational constants, and if the weighting functions were linear (such as \(w_1 s_1 + w_2 s_2\)), possibly including a min or a max, then we could use linear programming, and the decision procedure would go through in the presence of weights.

\section{Related work}\label{sec:related}
Rosser \cite{rosser60} comments on the possibility of considering formulas whose value is guaranteed to be at least $\alpha$.  For example, in
{\L}ukasiewicz) logic, if we consider weak disjunction  $\underline{\veebar}$, where $f_{\underline{\veebar}}(s_,s_2) = \max(s_1,s_2)$, then  the real value of 
$A \underline{\veebar} \neg A$ is always at least 0.5,  since $f_{\neg}(s) = 1-s$.
But Rosser rejects this approach, since he notes that there are uncountably many choices for $\alpha$, but only countably many recursively enumerable sets (and an axiomatization would give a recursively enumerable set  of valid formulas). 

Belluci \cite{belluce64}
investigates when the set of formulas with values at least $\alpha$ is recursively
enumerable.
Font et al.~\cite{fgtv06} consider the question of what they call ``preservation of degrees of truth''. 
They give a method for deciding, for a fixed $\alpha$,  if $\sigma$ having a value at least $\alpha$ implies that 
$\phi$ has value  at least  $\alpha$.

Nov{\'a}k \cite{Novak2015} considered a logic with sentences that assign a real value to each formula of first-order real-valued logic.  Thus, using our notation, his sentences would be of the form $(\phi; \{\alpha\})$, where $\phi$ is a formula in first-order real-valued logic, and $\alpha$ is a single real value.  He gave a sound and complete axiomatization. 
 
An interesting logic is the rational Pavelka logic RPL, an expansion of the standard {\L}ukasiewicz logicwhere rational truth-constants are allowed in formulas. For example, if $r$ is a rational number, then the formula $r \rightarrow \phi$ says that the value of $\phi$ is at least $r$, and the formula $\phi \rightarrow r$ says that the value of $\phi$ is at most $r$. Therefore, this logic can express the MD-sentences $(\phi; S)$, when $S$ is the union of a finite number of closed intervals. However, it cannot express strict inequalities.  For example, it cannot express that the value of $\phi$ is strictly greater than 0.5.\footnote{This follows from the stronger fact that if $A_1, \ldots, A_r$ are the atomic propositions, $\phi$ is a formula, and $G$ is the set of all value assignments to the atomic propositions that give $\phi$ the truth value 1,then since the operators of standard {\L}ukasiewicz logic are continuous (and so the value of $\phi$ is a continuous function of the value of the atomic propositions), it follows that $\{(g(A_1), \cdots, g(A_r)): \ g \in G\}$ is a closed subset of $[0,1]^r$. Note that if $r = 0.5$, then even though the formula $A \rightarrow r$ has the value 1 when the value $a$ of $A$ is at most 0.5,  the negation $\neg ( A \rightarrow r)$ does not have the value 1  when $a > 0.5$; instead it has the value $a$ - 0.5.} 
RPL was introduced by Hájek in \cite{hajek1998metamathematics} as a simplification of the system proposed by Pavelka in \cite{pavelka79} in which the syntax contained a truth-constant for each real number of the interval [0,1]. Hájek showed that an analogous logic could be presented as an expansion of {\L}ukasiewicz propositional logic with truth-constants only for the rational numbers in [0,1] and gave a corresponding completeness theorem. Moreover, first-order fuzzy logics with real or rational constants have also been deeply studied starting from Novák’s extension of Pavelka’s logic to a first-order predicate language in \cite{Novak:CompletenessFirstOrder} (see e.g.~\cite{Esteva-Godo-Noguera}).

Each of \cite{beavers1993automated}, \cite{mundici1994constructive} and \cite{hahnle1994many} give decision procedures that partially cover the situation we allow in Section~\ref{sec:decision}. The former two support  only {\L}ukasiewicz logic. The third, like our decision procedure, works for a variety of logics, though it is explicitly established in \cite{hahnle1994many} that their approach does not support discontinuous operators. Accordingly, unlike our decision procedure, their approach does not work for G{\"o}del logic given its discontinuous $\rightarrow$ operator.

\section{Conclusions}

We give a sound and strongly complete axiomatization for sentences about real-valued formulas. By being parameterized, our axiomatization covers essentially all real-valued logics. Our axiomatization allows us to include weights on formulas.
The results give us a way to establish such properties for neuro-symbolic systems that aim or purport to perform logical inference with real values.
The algorithm described gives us a constructive existence proof and a baseline approach for well-founded inference.
Because LNNs \cite{riegel2020logical} are exactly a weighted real-valued logical system implemented in neural network form, an important immediate upshot of our results for the weighted case is that they provide provably sound and complete logical inference for LNNs.
Such a result has not previously been established for any neuro-symbolic approach to our knowledge.
While our main motivation was to pave the way forward for neuro-symbolic systems, our results are fundamental, filling a long-standing gap in a very old literature, and can be applied well beyond AI.

\acknow{%
We are very grateful to Marco Carmosino, who improved the writing in this paper by giving  us many helpful comments
We are also grateful to Guillermo Badia, 
Ken Clarkson, Didier Dubois, Phokion Kolaitis,
Carles Noguera,
and Henri Prade for helpful comments.
Finally, we are grateful to Lluis Godo for confirming the novelty of our approach, and for helpful comments.
}

\showacknow{} %

\bibliography{bib}

\begin{thebibliography}{10}

\bibitem{serafini2016logic}
L Serafini, Ad Garcez, Logic tensor networks: Deep learning and logical
  reasoning from data and knowledge.
\newblock {\em\protect\JournalTitle{arXiv preprint}} \textbf{arXiv:1606.04422}
  (2016).

\bibitem{bach2017hinge}
SH Bach, M Broecheler, B Huang, L Getoor, Hinge-loss {Markov} random fields and
  probabilistic soft logic.
\newblock {\em\protect\JournalTitle{The Journal of Machine Learning Research}}
  \textbf{18}, 3846--3912 (2017).

\bibitem{cohen2020tensorlog}
W Cohen, F Yang, KR Mazaitis, {TensorLog}: A probabilistic database implemented
  using deep-learning infrastructure.
\newblock {\em\protect\JournalTitle{Journal of Artificial Intelligence
  Research}} \textbf{67}, 285--325 (2020).

\bibitem{riegel2020logical}
R Riegel, et~al., Logical neural networks.
\newblock {\em\protect\JournalTitle{arXiv preprint}} \textbf{arXiv:2006.13155}
  (2020).

\bibitem{boole1854investigation}
G Boole, {\em An investigation of the laws of thought: on which are founded the
  mathematical theories of logic and probabilities}.
\newblock (Walton and Maberly) Vol.{}~2, (1854).

\bibitem{zadeh1975fuzzy}
LA Zadeh, Fuzzy logic and approximate reasoning.
\newblock {\em\protect\JournalTitle{Synthese}} \textbf{30}, 407--428 (1975).

\bibitem{novak2008formal}
V Nov{\'a}k, A formal theory of intermediate quantifiers.
\newblock {\em\protect\JournalTitle{Fuzzy Sets and Systems}} \textbf{159},
  1229--1246 (2008).

\bibitem{epstein1993multiple}
G Epstein, {\em Multiple-valued logic design: an introduction}.
\newblock (CRC Press), (1993).

\bibitem{fagin2003optimal}
R Fagin, A Lotem, M Naor, Optimal aggregation algorithms for middleware.
\newblock {\em\protect\JournalTitle{Journal of computer and system sciences}}
  \textbf{66}, 614--656 (2003).

\bibitem{hajek1998metamathematics}
P H{\'a}jek, {\em Metamathematics of fuzzy logic}.
\newblock (Springer Science \& Business Media) Vol.{}~4, (1998).

\bibitem{rescher1969many}
N Rescher, {\em Many-valued logic}.
\newblock (McGraw-Hill), (1969).

\bibitem{rosser60}
JB Rosser, Axiomatization of infinite valued logics.
\newblock {\em\protect\JournalTitle{Logizue et Analyse}} \textbf{3}, 137--153
  (1960).

\bibitem{belluce64}
L Belluce, Further results on infinite valued predicate logic.
\newblock {\em\protect\JournalTitle{J. Symbolic Logic}} \textbf{29}, 69--78
  (1964).

\bibitem{fgtv06}
JM Font, {\`A}J Gil, A Torrens, V Verdu, On the infinite-valued {\l}ukasiewicz
  logic that preserves degrees of truth.
\newblock {\em\protect\JournalTitle{Arch. Math. Logic}} \textbf{45}, 835--868
  (2006).

\bibitem{Novak2015}
V Nov{\'a}k, Fuzzy logic with extended syntax.
\newblock {\em\protect\JournalTitle{Handbook of Mathematical Fuzzy Logic}}
  \textbf{3}, 1063--1104 (2015).

\bibitem{pavelka79}
J Pavelka, On fuzzy logic i, ii, iii.
\newblock {\em\protect\JournalTitle{Zeitschrift fur Mathematische Logik und
  Grundlagen der Mathematik}} \textbf{29}, 45—52, 119--134, 447--464 (1979).

\bibitem{Novak:CompletenessFirstOrder}
V Nov{\'{a}}k, On the syntactico-semantical completeness of first-order fuzzy
  logic part~{I} (syntax and semantic), part~{II} (main results).
\newblock {\em\protect\JournalTitle{Kybernetika}} \textbf{26}, 47--66, 134--154
  (1990).

\bibitem{Esteva-Godo-Noguera}
F Esteva, L God, C Noguerra, First-order t-norm based fuzzy logics with
  truth-constants: distinguished semantics and completeness properties.
\newblock {\em\protect\JournalTitle{Annals of Pure and Applied Logic}}
  \textbf{161}, 185--202 (2009).

\bibitem{beavers1993automated}
G Beavers, Automated theorem proving for {\l}ukasiewicz logics.
\newblock {\em\protect\JournalTitle{Studia Logica}} \textbf{52}, 183--195
  (1993).

\bibitem{mundici1994constructive}
D Mundici, A constructive proof of {McNaughton's} theorem in infinite-valued
  logic.
\newblock {\em\protect\JournalTitle{The Journal of Symbolic Logic}}
  \textbf{59}, 596--602 (1994).

\bibitem{hahnle1994many}
R H{\"a}hnle, Many-valued logic and mixed integer programming.
\newblock {\em\protect\JournalTitle{Annals of Mathematics and Artificial
  Intelligence}} \textbf{12}, 231--263 (1994).

\end{thebibliography}

\end{document}